\documentclass[11pt,a4paper]{article}
\pdfoutput=1
\usepackage[margin=1in]{geometry}
\usepackage[utf8]{inputenc}
\usepackage[dvipsnames]{xcolor}
\usepackage[bookmarksnumbered,linktocpage,hypertexnames=false,colorlinks=true,linkcolor=NavyBlue,urlcolor=NavyBlue,citecolor=ForestGreen,anchorcolor=green,breaklinks=true,pagebackref=true,pdfusetitle]{hyperref}
\usepackage{bbm,braket,microtype,amsmath,amssymb,amsthm,amsfonts,dsfont,dynkin-diagrams,mathtools,colortbl,booktabs,algorithm,algpseudocode,graphicx,enumitem,xspace,float,mathdots,ellipsis,caption,subcaption,mleftright,multirow,adjustbox}
\usepackage[T1]{fontenc}
\usepackage{textcomp}
\usepackage[english]{babel}
\usepackage[capitalize,nameinlink]{cleveref}
\usepackage{tikz}

\newtheorem*{theorem*}{Theorem}
\newtheorem{theorem}{Theorem}[section]\crefname{theorem}{Theorem}{Theorems}
\newtheorem{lemma}[theorem]{Lemma}\crefname{lemma}{Lemma}{Lemmas}
\crefname{claim}{Claim}{Claims}
\newtheorem{proposition}[theorem]{Proposition}\crefname{proposition}{Proposition}{Propositions}
\crefname{observation}{Observation}{Observations}
\newtheorem{corollary}[theorem]{Corollary}\crefname{corollary}{Corollary}{Corollaries}
\crefname{conjecture}{Conjecture}{Conjecture}
\theoremstyle{definition}
\newtheorem{definition}[theorem]{Definition}\crefname{definition}{Definition}{Definitions}
\newtheorem{problem}[theorem]{Problem}\crefname{problem}{Problem}{Problems}
\newtheorem{remark}[theorem]{Remark}\crefname{remark}{Remark}{Remarks}
\newtheorem{example}[theorem]{Example}\crefname{example}{Example}{Examples}
\crefname{condition}{Condition}{Conditions}
\numberwithin{equation}{section}

\usepackage{aliascnt}
\makeatletter
\let\c@algorithm\relax
\makeatother
\newaliascnt{algorithm}{theorem}

\makeatletter
\let\c@table\relax
\makeatother
\newaliascnt{table}{theorem}

\AddToHook{env/lemma/begin}{\crefalias{theorem}{lemma}}
\AddToHook{env/claim/begin}{\crefalias{theorem}{claim}}
\AddToHook{env/proposition/begin}{\crefalias{theorem}{proposition}}
\AddToHook{env/observation/begin}{\crefalias{theorem}{observation}}
\AddToHook{env/corollary/begin}{\crefalias{theorem}{corollary}}
\AddToHook{env/conjecture/begin}{\crefalias{theorem}{conjecture}}
\AddToHook{env/definition/begin}{\crefalias{theorem}{definition}}
\AddToHook{env/problem/begin}{\crefalias{theorem}{problem}}
\AddToHook{env/remark/begin}{\crefalias{theorem}{remark}}
\AddToHook{env/condition/begin}{\crefalias{theorem}{condition}}
\AddToHook{env/equation/begin}{\crefalias{theorem}{equation}}

\AddToHook{env/algorithm/begin}{\crefalias{theorem}{algorithm}}
\AddToHook{env/table/begin}{\crefalias{theorem}{table}}

\DeclareMathOperator{\diag}{diag}

\newcommand{\GL}{\mathrm{GL}}

\DeclareMathOperator{\maxrank}{maxrank}
\DeclareMathOperator{\minrank}{minrank}
\DeclareMathOperator{\rank}{rank}

\DeclareMathOperator{\SL}{SL}
\DeclareMathOperator{\spec}{spec}

\DeclareMathOperator{\supp}{supp}

\DeclareMathOperator{\Tr}{Tr}

\newcommand{\R}{\mathbb{R}}
\newcommand{\C}{\mathbb{C}}
\newcommand{\Z}{\mathbb{Z}}

\newcommand{\N}{\mathbb{N}}

\def\<#1>{\left\langle\ignorespaces#1\unskip\right\rangle}

\newcommand{\ot}{\otimes}

\DeclarePairedDelimiter\norm{\lVert}{\rVert}
\DeclarePairedDelimiter\floor{\lfloor}{\rfloor}

\DeclareMathOperator{\bordersubrank}{\underline{Q}}
\newcommand{\asympleq}{\lesssim}

\newcommand{\kron}{\boxtimes}

\newcommand{\Proj}{\mathbb{P}}

\newcommand{\tensor}{T} %
\newcommand{\sep}{\mid} %

\newcommand{\G}{\textnormal{GL}}
\newcommand{\Sl}{\textnormal{SL}}

\newcommand{\kronpol}[3]{\Delta(\C^{#1} \ot \C^{#2} \ot \C^{#3})}

\newcommand{\MM}{\mathsf{M}}
\newcommand{\IMM}[2]{\mathsf{M}_{#1}^{#2}}
\newcommand{\unit}[1]{\mathsf{U}_{\smash{#1}}}
\newcommand{\unitk}[2]{\mathsf{U}_{#1}^{#2}}
\newcommand{\polmul}{\mathsf{P}}
\newcommand{\degenleq}{\trianglelefteq}
\newcommand{\degengeq}{\trianglerighteq}
\allowdisplaybreaks[4]
\usepackage[normalem]{ulem}

\usepackage{authblk}

\begin{document}

\title{The moment polytope of matrix multiplication is not maximal}
\author[1,2]{Maxim van den Berg}
\author[3]{Matthias Christandl}
\author[1]{Vladimir Lysikov}
\author[3]{Harold Nieuwboer}
\author[1]{Michael Walter}
\author[2]{Jeroen Zuiddam}

\affil[1]{Ruhr University Bochum, Bochum, Germany}
\affil[2]{University of Amsterdam, Amsterdam, Netherlands}
\affil[3]{University of Copenhagen, Copenhagen, Denmark}

\date{}
\maketitle
\vspace{-2em}
\begin{abstract}
    Moment polytopes of tensors, the study of which is deeply rooted in invariant theory, representation theory and symplectic geometry, have found relevance in numerous places, from quantum information (entanglement polytopes) and algebraic complexity theory (GCT program and the complexity of matrix multiplication) to optimization (scaling algorithms). %
Towards an open problem in algebraic complexity theory, we prove separations between the moment polytopes of matrix multiplication tensors and unit tensors. As a consequence, we find that matrix multiplication moment polytopes are not maximal, i.e.\ are strictly contained in the corresponding Kronecker polytope. 
    As another consequence, we obtain a no-go result for a natural operational characterization of moment polytope inclusion in terms of asymptotic restriction. 
    We generalize the separation and non-maximality to moment polytopes of iterated matrix multiplication tensors. 
    Our result implies that tensor networks 
    where multipartite entanglement structures beyond two-party entanglement are allowed can go beyond projected entangled-pair states (PEPS) in terms of expressivity.

    Our proof characterizes membership of uniform points in moment polytopes of tensors, and establishes a connection to polynomial multiplication tensors via the minrank of matrix subspaces.
    As a result of independent interest, we extend these techniques to obtain a new proof of the optimal border subrank bound for matrix multiplication. %

\end{abstract}

\setcounter{tocdepth}{3}
\tableofcontents
\maketitle

\newpage
\section{Introduction}
\label{subsection:definitions and overview}

Moment polytopes of tensors, the study of which is deeply rooted in invariant theory, representation theory and symplectic geometry \cite{nessStratificationNullCone1984,brion1987momentMapImage}, have found relevance in numerous places, from algebraic complexity theory (in particular, the geometric complexity program \cite{burgisserGeometricComplexity2011, burgisserNonvanishingKroneckerCoefficients2011} and asymptotic spectra \cite{strassenKomplexitatUndGeometrie2005, christandlUniversalPointsAsymptotic2021, wigderson2022asymptotic}) and quantum information (quantum marginals and entanglement polytopes \cite{MR2197548, klyachkoQuantumMarginalProblem2004, christandl2007nonzeroKroneckerCoefficients, walterEntanglementPolytopes2013, sawicki2014momentumMapEntanglement, walter2014thesis}) to optimization (scaling algorithms \cite{burgisser2018tensorScaling,burgisserTheoryNoncommutativeOptimization2019}). These polytopes admit various descriptions, notably in terms of the types of irreducible representations that appear in the Schur--Weyl decomposition of powers of the tensor and in terms of quantum marginals of elements in the orbit closure of the tensor.

Bürgisser and Ikenmeyer \cite[Problem~7.3]{burgisserGeometricComplexity2011}, motivated by questions in algebraic complexity and quantum information, posed as a central open problem to determine the moment polytopes of the matrix multiplication tensors and the unit tensors (diagonal tensors). A construction of Bürgisser, Christandl and Ikenmeyer~\cite[Theorem~1]{burgisserNonvanishingKroneckerCoefficients2011} implies that the moment polytope of the unit tensor contains all points which are uniform on two of the three subsystems. However, no further progress has been made on this problem since then.

We prove separations between moment polytopes of matrix multiplication and unit tensors of varying sizes. As a consequence, we find that the moment polytope of matrix multiplication is not equal to the maximal moment polytope (i.e.~the Kronecker polytope). 
We extend this to separations and non-maximality for iterated matrix multiplication tensors, and derive implications for the expressivity of tensor network representations of quantum states (which has seen much recent interest, e.g. \cite{christandlTensorNetworkRepresentations2020a,christandl2023resourcetheorytensornetworks,landsburg2012geometryOfTensorNetworkStates, bernardi2022dimensionTensorNetwork}).

We summarize our main results here and discuss these in more detail in the rest of the paper:
\begin{itemize}
\item We prove a range of separations between moment polytopes of tensors, namely for matrix multiplication and unit tensors (diagonal tensors). These in particular imply strictness of inclusions of matrix multiplication moment polytopes in Kronecker polytopes.
This constitutes the first progress towards an open problem of Bürgisser and Ikenmeyer~\cite[Problem~7.3]{burgisserGeometricComplexity2011}, going beyond previous constructions of moment polytope points of Bürgisser, Christandl and Ikenmeyer \cite[Theorem~1]{burgisserNonvanishingKroneckerCoefficients2011}.

\item As a central ingredient for the above we further develop minrank of tensors and polynomial multiplication tensors. We use these ingredients to give a new proof of the optimal border subrank upper bound for matrix multiplication tensors that was previously obtained using geometric rank \cite{strassenRelativeBilinearComplexity1987, koppartyGeometricRankTensors2023a}.
\item As another consequence, we rule out an, \emph{a priori} natural, asymptotic characterization of moment polytope inclusion in terms of asymptotic restriction, which originates from the study of matrix multiplication algorithms \cite{strassenRelativeBilinearComplexity1987, christandlUniversalPointsAsymptotic2021}.

\item In the context of quantum information theory, as a consequence of the above, we show there exist joint marginals realizable by pure multipartite quantum states, which are not realizable (even approximately) by matrix product states (MPS) with certain bond dimensions.
This implies that projected entangled-pair states (PEPS) become more expressive when allowing genuine multipartite entanglement structures beyond the standard choice of two-party maximally entangled states.

\end{itemize}

\subsubsection*{Moment polytopes of tensors}

Let $V = \C^a \ot \C^b \ot \C^c$ be the space of $a\times b \times c$ tensors.
The product of general linear groups $\GL = \GL_a \times \GL_b \times \GL_c$ naturally acts on $V$ by local basis transformations.\footnote{We may leave $a,b,c$ implicit when clear from the context and just refer to $\GL$.}
To every $T \in V$ are naturally associated linear maps $T_1\colon \C^a \to \C^b \ot \C^c$,  $T_2\colon \C^b \to \C^a \ot \C^c$ and $T_3\colon \C^c \to \C^a \ot \C^b$ from which we obtain positive semidefinite matrices $T_1\!^* T_1 \in \C^{a \times a}$, $T_2\!^* T_2 \in \C^{b \times b}$ and $T_3\!^* T_3 \in \C^{c \times c}$. %
Let
\begin{align}
\label{definition:moment map marginal}
\mu_i\colon T \mapsto \frac{T_{i}\!^* T_i}{\Tr(T_i\!^*T_i)}.
\end{align}
The moment map $\mu\colon \C^a \ot \C^b \ot \C^c \setminus \{0\} \to \C^{a\times a} \times \C^{b \times b} \times \C^{c \times c}$ is defined by
\begin{align}
\label{definition:moment map}
\mu(T) \coloneqq \big(\mu_1(T), \mu_2(T), \mu_3(T)\big).
\end{align}
Note that $\Tr(T_i^*T_i) = \norm{T}^2$ does not depend on $i$.
For any positive semidefinite matrix $M \in \C^{n \times n}$, let $\spec(M) = (\lambda_1, \lambda_2, \ldots, \lambda_n) \in \R^n$, where $\lambda_1 \geq \lambda_2 \geq \cdots \geq \lambda_n \geq 0$, denote the eigenvalues of~$M$, non-increasingly ordered.
Define
\[
\spec(\mu(T)) \coloneqq \big(\spec(\mu_1(T)) \sep \spec(\mu_2(T)) \sep \spec(\mu_3(T))\big) \in \R^a \times \R^b \times \R^c,
\]
where $\sep$ separates the three components of $\R^a \times \R^b \times \R^c$.
\begin{definition}\label{definition:moment polytope}
For any irreducible algebraic variety $W \subseteq \C^a\ot\C^b\ot\C^c$ that is closed under the action of~$\GL$, the \emph{moment polytope} of $W$ is defined as
\[
\Delta(W) \coloneqq \big\{ \spec(\mu(S)) \mid S \in W\setminus \{0\} \big\} \subseteq \R^a \times \R^b \times \R^c.\,
\]
In the context of quantum information theory, the moment polytopes is also called the \emph{entanglement polytope} of $W$ \cite{walterEntanglementPolytopes2013}.
\end{definition}

Note that because $\GL$ can scale elements of $\C^a\ot\C^b\ot\C^c$, any algebraic variety $W$ that is closed under $\GL$ is an (algebraic) cone.
Because $\mu$ is invariant under scaling, we may equivalently work in projective space. This is the viewpoint of Ness, Mumford and Brion.
We state their result in our setting.
\begin{theorem}[{\cite{nessStratificationNullCone1984,brion1987momentMapImage, walterEntanglementPolytopes2013}\footnote{
To be able to apply the theorem of Mumford--Ness and Brion, we furthermore need to show show that irreducibility of a variety $W$ that is a cone is equivalent to irreducibility of the projective variety $\Proj(W) \subseteq \Proj(V)$ consisting of the lines through $W$, which is not hard to do.
}}]
\label{theorem:tensor moment polytopes nm}
    Let $W \subseteq \C^a\ot\C^b\ot\C^c$ be an irreducible algebraic variety that is closed under the action of $\GL$. Then $\Delta(W)$ is a (bounded convex) polytope with rational vertices.
\end{theorem}
For any tensor $T \in \C^a\ot\C^b\ot\C^c$, let $\GL\cdot T$ denote its $\GL$-orbit and $\overline{\GL\cdot T}$ the closure in the Euclidean topology (which is the same as the Zariski-closure\footnote{This is because the orbit is constructible \cite[Section~AG.1.3, Corollary~AG.10.2]{borel2012}.}). 
The orbit closure $\overline{\G \cdot T}$ is also irreducible\footnote{The group $\G$ is irreducible as it is connected \cite[Proposition~I.1.2]{borel2012}, which implies that $\overline{\G \cdot T}$ is irreducible as well.}.
The moment polytope of $T$ is then defined as
\[
\Delta(T) \coloneqq \Delta\bigl(\overline{\GL\cdot T}\bigr).
\]

There is a second description of the moment polytope via representation theory \cite{nessStratificationNullCone1984,brion1987momentMapImage,christandlUniversalPointsAsymptotic2021,burgisser2018tensorScaling} and a third description in terms of achievable supports of the tensor under the action of lower triangular matrices \cite{franz2002}. The representation-theoretic characterization exhibits the moment polytope as the set of normalized highest weights of $\G$ whose irreducible representations in $(\C^a\ot\C^b\ot\C^c)^{\ot n}$ have non-zero overlap with $T^{\ot n}$ for some $n > 0$.

\subsubsection*{Kronecker polytope and moment polytopes for matrix multiplication and unit tensors}
The moment polytope $\Delta(\C^a\ot\C^b\ot\C^c)$ for the whole space $\C^a \ot \C^b \ot \C^c$ is called the \emph{Kronecker polytope}. It has an alternative description in terms of the Kronecker coefficients for the symmetric group. For every tensor $T \in \C^a\ot\C^b\ot\C^c$, we have  $\Delta(T) \subseteq \Delta(\C^a\ot\C^b\ot\C^c)$ and it can be shown that for generic~$T$, we have $\Delta(T) = \Delta(\C^a\ot\C^b\ot\C^c)$.%
\footnote{That is, there is a non-empty Zariski-open subset $U \subseteq \C^a\ot\C^b\ot\C^c$ such that for every $T \in U$, $\Delta(T) = \Delta(\C^a\ot\C^b\ot\C^c)$. In particular, it holds for a random tensor $T$ with probability one.}
Determining $\Delta(\C^a \ot \C^b \ot \C^c)$ is a hard problem with a long history; we give an overview below.

Motivated by geometric complexity theory (and in particular the study of the matrix multiplication exponent), Bürgisser and Ikenmeyer \cite{burgisserGeometricComplexity2011} posed the problem of determining the moment polytopes of the unit tensors and matrix multiplication tensors.
For $r \in \N$, the unit tensor of rank $r$ is defined as $\unit{r} \coloneqq \sum_{i=1}^r e_i \ot e_i \ot e_i$.
For $n_1,n_2,n_3\in \N$, the matrix multiplication tensors are defined as
\[
  \MM_{n_1,n_2,n_3} \coloneqq \sum_{i=1}^{n_1} \sum_{j=1}^{n_2} \sum_{k=1}^{n_3} e_{i,j} \ot e_{j,k} \ot e_{k,i} \ \in\  \C^{n_1 n_2} \ot \C^{n_2 n_3} \ot \C^{n_3 n_1},
\]
where $e_{i,j}$ denote standard basis vectors. We take $e_{i,j}$ to equal the standard basis matrix (of the correct size) with a 1 at position $(i,j)$, flattened to a vector along the rows.
We define $\MM_{n} = \MM_{n,n,n} \in \C^{n^2} \otimes \C^{n^2} \otimes \C^{n^2}$.
\begin{problem}[\cite{burgisserGeometricComplexity2011}]
Determine $\Delta(\unit{n})$ and $\Delta(\MM_{n})$.
\end{problem}

There is much previous work on characterizations of the Kronecker polytope in various formats. We provide a brief overview. See also \cite[Chapter~3]{walter2014thesis} for a detailed account.
Various complete mathematical descriptions of the Kronecker polytope of arbitrary formats have been found using different techniques \cite{nessStratificationNullCone1984,berenstein2000coadjointOrbitsMomentPolytope,franz2002,ressayre2010generalizedEigenvalueProblem,ressayre2011generalizedEigenvalueProblemII,vergneInequalitiesMomentCone2017}.
They, among other techniques, have been used to determine explicit descriptions of the inequalities (i.e.\ as concrete lists of numbers) for the following formats: $3 \times 3 \times 3$ \cite{franz2002}, $2 \times 2 \times 4$ \cite{bravyi2004requirementsQuantumStates}, $2 \times 2 \times \cdots \times 2$ \cite{higuchi2003qubitReducedStates},
$2 \times 2 \times 2 \times 8$,$\ $
$2 \times 2 \times 2 \times 2 \times 16$,$\ $
$3 \times 3 \times 9$, $\ $
$2 \times n \times 2n$,$\ $
$2\times 2\times 3 \times 12$ \cite{klyachkoQuantumMarginalProblem2004} (based on \cite{berenstein2000coadjointOrbitsMomentPolytope} and the connection to Kronecker coefficients),
and $4 \times 4 \times 4$ \cite{vergneInequalitiesMomentCone2017} (using techniques related to those in \cite{ressayre2010generalizedEigenvalueProblem}).
For all results, the formats with system-wise lesser or equal dimensions can also be obtained, as can permutations of the formats.

For moment polytopes $\Delta(T)$ of specific tensors $T$ much less is known.
Bürgisser, Christandl and Ikenmeyer \cite{burgisserNonvanishingKroneckerCoefficients2011} used methods from quantum information theory to construct a large class of points in the Kronecker polytopes and unit tensor moment polytopes.
For $m \in \N$ let $u_m = (1/m,\ldots, 1/m) \in \R^m$ be the uniform probability vector of length $m$.
They proved that for any non-increasing probability vector $q = (q_1,q_2,\ldots, q_{n^2})$ the point $(q \sep u_n \sep u_n)$ is in $\Delta(\C^{n^2} \ot \C^n \ot \C^n)$. In fact, their construction gives that for any non-increasing probability vector $q = (q_1, q_2, \ldots, q_n)$ the point $(q \sep u_n \sep u_n)$ is in $\Delta(\unit{n})$.%
\footnote{The statement of Theorem~1 in~\cite{burgisserNonvanishingKroneckerCoefficients2011} is about the Kronecker polytope for format~$(n^2,n,n)$, but their construction for points in~$(n,n,n)$ only requires tensors of rank~$\leq n$.
Given a probability distribution~$q$ on~$n$ outcomes, the point~$p = (q \sep u_n \sep u_n)$ arises from the tensor~$T = \sum_{k=1}^n (\sum_{j=1}^n \sqrt{q_j} \zeta_n^{jk}) e_j \ot e_k \ot e_k$, where~$\zeta_n$ is a primitive~$n$-th root of unity. This tensor~$T$ in the~$\G$-orbit closure of~$\<n>$, as it has rank at most~$n$.}
The moment polytopes of every tensor of format $2 \times 2 \times 2$ and $2 \times 2 \times 2 \times 2$ were determined in
\cite{han2004compatibleConditionsEntanglement,sawicki2013threeQubits} and \cite{walterEntanglementPolytopes2013} respectively.

Using the tensor scaling algorithm from \cite{burgisser2018tensorScaling}, there are numerical methods to determine whether $\Delta(T) = \Delta(\C^a\ot\C^b\ot\C^c)$, whenever $\Delta(\C^a\ot\C^b\ot\C^c)$ is known.
More precisely, tensor scaling algorithms decide the membership problem: whether a given point $p$ is in $\Delta(T)$. Although not efficient in general, yes-instances are determined fast in practice. Letting $p$ range over the vertices of $\Delta(\C^a\ot\C^b\ot\C^c)$, we can determine whether they also lie in $\Delta(T)$. 

\subsubsection*{Moment polytope separations} %

Our first result is a range of separations between moment polytopes of matrix multiplication and unit tensors. Recall that $u_n \in \R^n$ denotes the uniform probability vector. We implicitly pad this vector with zeroes so that $u_n \in \R^m$ for any $m \geq n$ is the uniform probability vector on the first~$n$ coefficients.
\begin{theorem}
\label{theorem:matrix multiplication}
For every $n \in \N$, $\Delta(\unit{n^2}) \not\subseteq \Delta(\MM_{n})$.
More generally, for every $c,n \in \N$, if $n^2-n+1 < c \leq n^2$, then $\Delta(\unit{c}) \not\subseteq \Delta(\MM_{n})$.
Namely, the point $p_c\coloneq( u_2 \sep u_{c-1} \sep u_c )$ satisfies $p_c \in \Delta(\unit{c})$ and $p_c \notin \Delta(\MM_{n})$. %
\end{theorem}

From symmetries of the tensors $\unit{c}$ and $\MM_{n}$, it follows from \autoref{theorem:matrix multiplication} that for every $\pi \in S_3$ we have that $\pi \cdot p_c \in \Delta(\unit{c})$ and $\pi \cdot p_c \not\in \Delta(\MM_{n})$, where the symmetric group $S_3$ naturally acts on $\R^{n^2}\times\R^{n^2}\times\R^{n^2}$ by permuting the three components.%
\footnote{The argument is as follows.
The symmetric group $S_3$ acts on any tensor space $\C^a \otimes \C^b \otimes \C^c$ by permuting the factors, and similarly acts on any moment polytope by permuting the three components of its ambient space $\R^a \times \R^b \times \R^c$.
The unit tensor $\unit{c}$ is $S_3$-invariant.
The matrix multiplication tensor $\MM_{n}$, although not $S_3$-invariant, has the property that $\pi \cdot \MM_{n} \in \GL\cdot\MM_{n}$ for every $\pi \in S_3$.
As a result, the moment polytopes $\Delta(\unit{c})$ and $\Delta(\MM_{n})$ are $S_3$-invariant.
From \cref{theorem:matrix multiplication} it thus follows that for every $\pi \in S_3$ we have $\pi \cdot p_c \in \Delta(\unit{c})$ and $\pi \cdot p_c \not\in \Delta(\MM_{n})$.}

\cref{theorem:matrix multiplication} in particular implies that the matrix multiplication tensor $\MM_{n}$ does not have maximal moment polytope, as follows. %
\begin{corollary}
    The inclusion $\Delta(\MM_{n}) \subseteq \kronpol{n^2}{n^2}{n^2}$ is a strict inclusion.
\end{corollary}
\begin{proof}
$\Delta(\unit{n^2})$ and $\Delta(\MM_{n})$ are both contained in $\kronpol{n^2}{n^2}{n^2}$. From \cref{theorem:matrix multiplication} we have $\Delta(\unit{n^2}) \not\subseteq \Delta(\MM_{n})$, so the claim follows.
\end{proof}
\cref{theorem:matrix multiplication} leaves open whether the inclusion $\Delta(\MM_{n}) \subseteq \Delta(\unit{n^2})$ holds.
More generally, while it is known that $\kronpol{m}{m}{m} = \Delta(\unit{m})$ for every $m \in \{2,3,4\}$, this equality is open for all larger $m$.
If true, then $\Delta(\MM_{n}) \subsetneq \Delta(\unit{n^2})$.
The separating point in \cref{theorem:matrix multiplication} for $n=2,c=4$ we first observed using an algorithm that we developed to compute moment polytopes \cite{vandenBerg2025momentPolytopeAlgorithm}, which in turn inspired the theorem.
Namely, we found computationally
that $p_4 = (\tfrac12,\tfrac12,0,0 \sep \frac13,\frac13,\frac13,0 \sep \frac14, \frac14, \frac14, \frac14 )$, which is an element of $\kronpol{4}{4}{4}$,
is not an element of $\Delta(\MM_2)$. Moreover, $\Delta(\unit{4})$ equals $\kronpol{4}{4}{4}$, as can be seen using the tensor scaling algorithm \cite{burgisser2018tensorScaling} and knowledge of the vertices of $\kronpol{4}{4}{4}$, which were determined in \cite{vergneInequalitiesMomentCone2017}.

Our approach to proving \cref{theorem:matrix multiplication} is to show that $\MM_{n}$ cannot degenerate to polynomial multiplication tensors of certain shapes.
For shape $2 \times (c-1) \times c$, such a degeneration is required to have $(u_2 \sep u_{c-1} \sep u_c)$ as an element of the moment polytope.
We achieve this by upper bounding the \emph{minrank} of $\MM_{n}$ under degenerations. Minrank is defined as the smallest rank among the non-zero matrices in the slice span.

\subsubsection*{Higher-order tensors}

We extend \cref{theorem:matrix multiplication} to iterated matrix multiplication tensors and unit tensors of order~$k$, for any $k\geq 3$. The definition of moment polytopes generalizes naturally to the space of tensors of order $k$, $V = \C^{n_1} \ot \cdots \ot \C^{n_k}$, with the action of $\GL = \GL_{n_1} \times \cdots \times \GL_{n_k}$.
Let $\unitk{r}{k} = \sum_{j=1}^{r} e_{j} \ot \dotsb \ot e_{j} \in (\C^r)^{\ot k}$ denote the unit tensor of rank $r$ and order $k$.
Let $\IMM{n}{k} \in (\C^{n^2})^{\ot k}$ be the iterated matrix multiplication tensor of order $k$ for $n \times n$ matrices, which is defined as
\[
\IMM{n}{k} \coloneqq \sum_{i \in [n]^k} \bigotimes_{\ell=1}^{k} e_{i_{\ell},i_{\ell+1}},
\]
where we set $i_{k+1}=i_1$.

\begin{theorem}\label{cor:it-mat}
    \label{cor:IMM does not have generic polytope}
    For every $n$ and $c$ such that $n^2 - n + 1 < c \leq n^2$, 
    we have that $\Delta(\unitk{c}{k}) \not\subseteq \Delta(\IMM{n}{k})$.
\end{theorem}

In the language of tensor networks,
\cref{cor:IMM does not have generic polytope} can be interpreted as follows.

\begin{corollary}
    For~$k \geq 3$ parties and $n,c \in \N$ such that~$n^2 - n + 1 < c \leq n^2$, there are marginals that can be realized by applying local maps to a GHZ state with $c$ levels, but are not realizable as a matrix product state with bond dimension $n$ and periodic boundary conditions. 
\end{corollary}
In the language of \cite{christandlTensorNetworkRepresentations2020a,christandl2023resourcetheorytensornetworks}, a matrix product state with periodic boundary conditions is precisely a projected entangled-pair state (PEPS) on the $k$-cycle graph $C_k$.
This entanglement structure is described by pairwise level-$n$ maximally entangled states for every edge of~$C_k$ ($\IMM{n}{k}$).
We may replace this by the entanglement structure described by the hypergraph on $k$ vertices with a single hyper-edge containing all the vertices, which is the level-$c$ GHZ state shared between all parties ($\unitk{c}{k}$).
Then the above means that this replacement alters the expressivity (on the level of joint one-body marginal spectra) of the tensor network whenever $n^2 - n + 1 < c \leq n^2$, in particular showing that allowing multipartite entanglement structures beyond two-party entanglement can increase expressiveness in tensor networks.
Analogous separations can be derived for different graphs and hypergraphs governing tensor networks.

\subsubsection*{Degeneration obstructions and border subrank of matrix multiplication}

For any two tensors $S$ and $T$ we say that $S$ is a restriction of $T$ and write $S \leq T$ if there are linear maps $A,B,C$ such that $S = (A\otimes B \otimes C)\cdot T$. We say $S$ is a degeneration of $T$ and write $S \degenleq T$ if there is a sequence of tensors $T_i$ that converges to $S$ and such that $T_i \leq T$ for every~$i$.

Moment polytope separations, like the separations $\Delta(\unit{c}) \not\subseteq \Delta(\MM_{n})$ for $n^2-n+1 < c \leq n^2$ that we proved in \cref{theorem:matrix multiplication}, are obstructions (in the spirit of geometric complexity theory~\cite{burgisserGeometricComplexity2011}) for degenerations. Indeed, if $S \degenleq T$, then $\Delta(S) \subseteq \Delta(T)$ (\cref{proposition:degeneration monotone}). Thus, if $\Delta(S) \not\subseteq \Delta(T)$, then $S \not\degenleq T$. From \cref{theorem:matrix multiplication} we thus get that $\unit{n^2 - n + 2} \not\degenleq \MM_{n}$.

Statements of the form $\unit{r} \not\degenleq T$ correspond to upper bounds on Strassen's notion of border subrank \cite{strassenRelativeBilinearComplexity1987}, which plays a central role in the study of matrix multiplication algorithms \cite{blaser2013fast}.
The border subrank $\bordersubrank(T)$ of a tensor $T$ is defined as the largest number $r$ such that $\unit{r} \degenleq T$, so that $\unit{r} \not\degenleq T$ corresponds to $\bordersubrank(T) < r$. Strassen proved that $\lceil\tfrac34n^2\rceil \leq \bordersubrank(\MM_{n})$. This was shown to be an equality: %
\begin{theorem}[\cite{koppartyGeometricRankTensors2023a}]
\label{theorem:subrank bound}
$\bordersubrank\big(\MM_{n}\big) \leq \lceil \tfrac34 n^2\rceil$.
\end{theorem}

From our moment polytope separation (\cref{theorem:matrix multiplication}) it follows that $\bordersubrank(\MM_{n}) \leq n^2 - n + 1$, which is not optimal.
We slightly alter the proof of the moment polytope separation to reprove the optimal upper bound $\bordersubrank(\MM_{n}) \leq \lceil\tfrac34n^2\rceil$,
obtaining a proof that is more direct than the previous proof of \cite{koppartyGeometricRankTensors2023a}, shedding new light on this result. %
We leave as an open problem whether a moment polytope separation can prove the optimal bound.

\subsubsection*{Characterizing moment polytope inclusion}

What is the operational meaning of moment polytope inclusion? If $S \degenleq T$, then $\Delta(S) \subseteq \Delta(T)$ (\cref{proposition:degeneration monotone}). The reverse implication is known to be false (as can be seen for instance by considering non-equivalent generic tensors in $\C^3 \ot \C^3 \ot \C^3$).
What preorders on tensors imply moment polytope inclusion?

We consider asymptotic restriction, a preorder that is implied by degeneration (i.e.\ a larger preorder) that plays a central role in algebraic complexity theory. %
For tensors $S$ and $T$ we say that $S$ is an asymptotic restriction of $T$, denoted by $T \lesssim S$, if for every $n$ we have $S^{\kron n} \leq T^{\kron (n + o(n))}$ where $o(n)$ denotes some function $f(n)$ such that $f(n)/n \to 0$ when $n\to\infty$. Replacing restriction by degeneration does not change the notion of asymptotic restriction \cite[Prop.~5.10]{strassenRelativeBilinearComplexity1987}.
Here~$\kron$ denotes the Kronecker product on 3-tensors.

We show that asymptotic restriction does not imply moment polytope inclusion:
\begin{theorem}
    \label{thm:moment-polytope-not-asymptotic-restriction-monotone}
    There exist tensors~$S, T$ such that~$S \lesssim T$ and~$\Delta(S) \not\subseteq \Delta(T)$.
\end{theorem}
We prove \cref{thm:moment-polytope-not-asymptotic-restriction-monotone} by providing two counterexamples.
Both were found via to the aforementioned algorithm for computing moment polytopes \cite{vandenBerg2025momentPolytopeAlgorithm} (in fact, the first example uses the separation from \cref{theorem:matrix multiplication}).

From the perspective of quantum information theory, \cref{thm:moment-polytope-not-asymptotic-restriction-monotone} implies that existence of asymptotic SLOCC interconversion is not an entanglement monotone relation, in the sense that $S \lesssim T$ does not imply that all marginal spectra that can be reached starting from the state $S$ (using SLOCC transformations) can also be reached starting from the state $T$. 

In algebraic complexity theory, the above result is especially interesting in light of a family of $\lesssim$-monotone functions mapping tensors to $\R_{\geq 0}$ called the \emph{quantum functionals} \cite{christandlUniversalPointsAsymptotic2021}.
These functions are defined as an entropy maximizations over the moment polytope $\Delta(T)$.
The quantum functionals comprise all known elements of the \emph{asymptotic spectrum} of 3-tensors, which is the set $\mathcal X$ of functions mapping 3-tensors to $\R_{\geq 0}$ that are monotone under $\geq$, additive under direct sum, multiplicative under the Kronecker product, and map $\unit{r}$ to $r$.
Strassen showed that $T \lesssim S$ if and only if $\varphi(T) \leq \varphi(S)$ for all functions $\varphi \in \mathcal X$, and that the asymptotic tensor rank of $T$ equals the maximum of $\varphi(T)$ over $\varphi \in \mathcal X$ \cite{strassen1988asymptotic}.
It is a important open question to determine all of $\mathcal X$ explicitly.
It is also open whether $\mathcal X$ consists of just the quantum functionals. If true, this would in particular imply that the matrix multiplication exponent is equal to 2.

A natural question to ask is: is the map $T \mapsto \Delta(T)$ itself monotone under $\lesssim$, or do the quantum functionals pick out only some special information from $\Delta(T)$?
\Cref{thm:moment-polytope-not-asymptotic-restriction-monotone} shows that the latter is the case.
\section{Preliminaries}
\label{section:tensor moment polytopes}

In this section we recall background material on geometric invariant theory and basic properties of moment polytopes.
We continue where we left off in the introduction.

\subsection{Semistability, polystability, and invariants}
\label{subsection:geometric invariant theory}

We will need some basic concepts from geometric invariant theory (for an introduction, see e.g.\ \cite{derksen2013computational,wallachGeometricInvariantTheory2017}).
Let $\Sl_n \subseteq \GL_n$ denote the special linear group consisting of $n \times n$-matrices with determinant one.
Let $\Sl \coloneqq \SL_a \times \SL_b \times \SL_c \subseteq \G$. %
We call a tensor $T \in \C^a \ot \C^b \ot \C^c$ \emph{$\Sl$-unstable} when $0 \in \overline{\Sl \cdot T}$, and \emph{$\Sl$-semistable} otherwise. The subset of $\Sl$-unstable tensors is called the \emph{null cone}.
If $V$ is any (rational) $\SL$-representation, then $\SL$ acts on the ring of polynomials $\C[V]$ by $g\cdot f = f \circ g^{-1}$. We denote by $\C[V]^{\Sl}$ the $\Sl$-invariant polynomials on $V$.

\begin{theorem}[{\cite[Lemma~2.5.2]{derksen2013computational}}]
    \label{proposition:tensor null-cone}
    Let $V$ be a rational $\Sl$-representation.
    The subset of $\Sl$-unstable elements of $V$%
    is equal to the subset of elements $v \in V$ such that $f(v) = 0$ for every non-constant homogeneous $f \in \C[V]^{\Sl}$. %
\end{theorem}

We now turn our attention to the Kempf--Ness theorem,
which relates $\Sl$-stability to the moment map.
Recall the definition of the moment map $\mu$ as given in \cref{definition:moment map}.
Denote with $\<\cdot,\cdot>$ and $\norm{\cdot}$ the standard Euclidean inner product and norm on $\C^a\ot\C^b\ot\C^c$.
A tensor such that $\SL \cdot T$ is closed is called \emph{$\Sl$-polystable} (this also implies $T$ is $\Sl$-semistable if $T$ is non-zero).
Denote with $I_a$ the $a \times a$ identity matrix.
\begin{theorem}[Kempf--Ness theorem {\cite[Theorem~3.26]{wallachGeometricInvariantTheory2017}}\footnotemark{}]
    \label{theorem:tensor kempf-ness}
    Let $T \in \C^a\ot\C^b\ot\C^c \setminus \{0\}$.
    The tensor $T$ is $\Sl$-polystable if and only if $\mu(S) = (I_a/a,I_b/b,I_c/c)$ for some $S \in \Sl \cdot T$.
\end{theorem}
\footnotetext{
To apply the statement as given in \cite[Theorem~3.26]{wallachGeometricInvariantTheory2017}, we need to show that $\mu(T) = (I_a/a,I_b/b,I_c/c)$ is equivalent to $T$ being ``critical''. In our setting, this translates to the requirement that
$\<(A \ot I \ot I) T, T> + \<(I \ot B \ot I) T, T> + \<(I \ot I \ot C) T, T> = 0$ for all Hermitian matrices $A,B,C$ each of trace zero. A simple computation shows that these two requirements are equivalent.
For more information, see \cite[Example~1.4]{burgisserTheoryNoncommutativeOptimization2019} \cite{burgisser2018tensorScaling,vandenBerg2025momentPolytopeAlgorithm}.
}

Using the Kempf--Ness theorem (\cref{theorem:tensor kempf-ness}) combined with \cref{proposition:tensor null-cone}, we can prove that whenever a tensor $T$ has a tensor with uniform marginals in its orbit closure, it is $\SL$-semistable, as was also observed in \cite[Theorem~3.2]{burgisser2018alternatingMinimization} and \cite[Lemma~4.34]{christandlUniversalPointsAsymptotic2021}.\footnote{In \cite{burgisser2018alternatingMinimization} a self-contained proof (without assuming the Kempf--Ness theorem) is also provided.}

\begin{theorem}[Uniform marginals {\cite[Theorem~3.2]{burgisser2018alternatingMinimization}}]
    \label{corollary:uniform marginals}
    Let $T \in \C^a\ot\C^b\ot\C^c \setminus \{0\}$.
    The tensor $T$ is $\Sl$-semistable if and only if $\mu(S) = (I_a/a,I_b/b,I_c/c)$ for some non-zero $S \in \overline{\G \cdot T}$.
\end{theorem}

In other words, \cref{corollary:uniform marginals} says that $T$ is $\Sl$-semistable if and only if $(u_a,u_b,u_c) \in \Delta(T)$, where $u_n = (1/n,\ldots, 1/n) \in \R^n$ is the uniform probability vector of length $n$.

For more information on these theorems, we refer to \cite{burgisserTheoryNoncommutativeOptimization2019,franks2022minimallengthorbitclosure,vandenBerg2025momentPolytopeAlgorithm}.

\subsection{Moment polytopes under restriction and degeneration}

We review some basic properties of moment polytopes that we will use throughout the text.

We say a tensor $T$ \emph{restricts} to a tensor $T' \in \C^{a'} \ot \C^{b'} \ot \C^{c'}$ whenever there exist triples of matrices $(A,B,C) \in \C^{a' \times a} \times \C^{b' \times b} \times \C^{c' \times c}$ (not necessarily invertible) such that $(A \ot B \ot C)T = T'$. We then write $T \geq T'$.
When $T \geq T'$ and $T' \geq T$, we say $T$ and $T'$ are equivalent and write $T \sim T'$.
Whenever $T' \in \C^a\ot\C^b\ot\C^c$, we have that $T \sim T'$ if and only if $T' \in \G \cdot T$, in which case we say $T$ and $T'$ are isomorphic.
We say $T$ \emph{degenerates} to $T'$ if $\overline{\G \cdot T}$ contains a tensor equivalent to $T'$, and write $T \degengeq T'$.
Restriction implies degeneration, and both are transitive relations.
We say a tensor is \emph{concise} whenever it can not be embedded into a smaller space.
That is, $T$ is concise when it is not equivalent to a tensor $T' \in \C^{a'}\ot\C^{b'}\ot\C^{c'}$ with $a' < a$ or $b' < b$ or $c' < c$.

Moment polytopes behave well under equivalence of tensors.
First we show that embedding~$T$ into a larger space does not change the moment polytope apart from padding it with zeros.
Denote with $0_{x,y,z} \in \C^x \ot \C^y \ot \C^z$ the zero tensor.
Given $p = (p_1 \sep p_2 \sep p_3) \in \Delta(T)$, we denote the padding of $p$ with zeros on each system by
\begin{equation}
    (p_1 \sep p_2 \sep p_3) \oplus (0_x \sep 0_y \sep 0_z) \coloneqq (p_1 \oplus 0_x \sep p_2 \oplus 0_y \sep p_3 \oplus 0_z) \in \R^{a+x}\times\R^{b+y}\times\R^{c+z}.
\end{equation}

\begin{lemma}
\label{lemma:padding with zeros}
$\Delta(T \oplus 0_{x,y,z}) = \Delta(T) \oplus (0_x \sep 0_y \sep 0_z)$.
\end{lemma}
\begin{proof}

    $\subseteq$.
    Write $S = T \oplus 0_{x,y,z}$.
    Let $q = (q_1,q_2,q_3) \in \Delta(S)$.
    Then there is an element $S' \in \overline{\G \cdot S}$ with $\spec \mu(S') = q$.
    By applying some element in $\G$, we may assume that $\mu_i(S') = \diag(q_i)$ for every~$i$. 
    Then $q = p \oplus (0_x \sep 0_y\sep 0_z)$ for some $p \in \R^a \times \R^b \times \R^c$:
    to see this, note that each $q_i$ is non-increasing, and that $\rank \mu_i(S') = \rank S_i' \leq \rank S_i$.  
    Now let $P_1 \colon \C^{a \times (a+x)}$ be the projection onto the first $a$ coordinates. Similarly define $P_2$ and $P_3$.
    Then $(P_1 \ot P_2 \ot P_3) S' \in \overline{\G \cdot T}$, and we find that 
    $\mu_i ((P_1 \ot P_2 \ot P_3) S') = P_i\mu_i(S') P_i^* = \diag(p_i)$ for every $i$, so $p \in \Delta(T)$.

    $\supseteq$. %
    We have $(\G_a \times \G_b \times \G_c \cdot T) \oplus 0_{x,y,z} \subseteq \G_{a+x}\times\G_{b+y}\times\G_{b+z} \cdot (T\oplus 0_{x,y,z})$. The inclusion follows.
\end{proof}

\begin{proposition}%
    \label{proposition:moment polytope embedding}
    Let $T$ and $T'$ be equivalent tensors.
    Then $\Delta(T)$ and $\Delta(T')$ are equal up to padding with zeros.
\end{proposition}
\begin{proof}
Pad $T$ and $T'$ with zeros such that the resulting tensors $\hat T$ and $\hat T'$ both lie in the same tensor space.
Then $\hat T$ and $\hat T'$ are still equivalent tensors, and because they live in the same tensor space we have $\hat T \in \G \cdot \hat T'$. This implies that $\Delta(\hat T) = \Delta(\hat T')$.
The result is then given by \cref{lemma:padding with zeros}.
\end{proof}

\begin{remark}
\label{remark:tensor moment polytope embedding}
By \cref{proposition:moment polytope embedding}, we can write statements such as $\Delta(T) = \Delta(T')$, $\Delta(T) \supseteq \Delta(T')$ and $p \in \Delta(T)$ even when $T$ and $T'$ are tensors of different dimensions and $p$ is a vector not in $\R^a \times \R^b \times \R^c$.
We understand the statements to concern the polytopes and points after appropriate padding or removal of zeros on each of the three systems, such that both live in the same space.

Importantly, the notion of $\Sl$-semistability does depend on the embedding of the tensor $T$ (in fact, $\Sl$-semistable tensors must always be concise).
So in contrast to the above, statements about $\Sl$-semistability will always be with respect to a fixed format.
\end{remark}

Degeneration implies inclusion of moment polytopes:

\begin{proposition}%
\label{proposition:degeneration monotone}
If $T \degengeq T'$ then $\Delta(T) \supseteq \Delta(T')$.
\end{proposition}
\begin{proof}
Degeneration is a transitive relation, so after embedding in a large enough space, it follows that $\overline{\G \cdot T} \supseteq \overline{\G \cdot T'}$.
The result then follows directly from the definition of the moment polytope (\cref{definition:moment polytope}).
\end{proof}

Moreover, almost all restrictions of a tensor $T$ of shape $(x,y,z)$ to a tensor $T'$ of shape $(a,b,c)$ will result in a moment polytope $\Delta(T')$ given by intersecting $\Delta(T)$ with $\R^{a} \times \R^{b} \times \R^{c}$.

\begin{proposition}[{\cite[Corollary~3.7]{burgisser2018tensorScaling}}]
    \label{proposition:tensor moment polytope generic restriction}
    Let~$T \in \C^x \ot \C^y \ot \C^z$ and $(a,b,c) \leq (x,y,z)$ entry-wise.
    Denote with $P_a \ot P_b \ot P_c$ the restriction from $\C^x\ot\C^y\ot\C^z$ to $\C^{a} \ot \C^{b} \ot \C^{c}$ that projects onto the first $a,b$ and $c$ coordinates in the respective system.
    Then for generic $(A,B,C) \in \G$,
    \[
        \Delta\big( (P_a \ot P_b \ot P_c)(A \ot B \ot C)T \big) = \Delta(T) \cap \big(\R^{a} \times \R^{b} \times \R^{c} \big).
    \]
    In particular, there exists a restriction $T' \in \C^a \ot \C^b \ot \C^c$ of $T$ such that $\Delta( T' ) = \Delta(T) \cap (\R^{a} \times \R^{b} \times \R^{c})$.
\end{proposition}

\section{Moment polytope separation for matrix multiplication and unit tensors}\label{subsec:mom-sep}

We will now work towards the proof of \cref{theorem:matrix multiplication}.

\subsection{Uniform points in the moment polytope and semistability}
We will show that inclusion of any ``uniform'' point $(u_a \sep u_b \sep u_c)$ in the moment polytope of a tensor $T$ is characterized by the existence of a restriction of $T$ to an $\SL$-semistable tensor of shape $a \times b \times c$.
We will then use this to prove \cref{theorem:matrix multiplication} by showing that $\MM_{n}$ cannot restrict to any $\Sl$-semistable tensor in $\C^2 \ot \C^{c-1} \ot \C^c$.

\begin{lemma}[Uniform points and semistability]
    \label{lemma:tensor semistability}
    Let $1 \leq a \leq a'$, $1 \leq b \leq b'$, $1\leq c \leq c'$ be integers. Let $\tensor \in \C^{a'} \ot \C^{b'} \ot \C^{c'}$.
    The following are equivalent: 
    \begin{enumerate}[label=\upshape(\arabic*)]
        \item \label{item:semistable}
        $(u_a \sep u_b \sep u_c) \in \Delta(\tensor)$.
        \item \label{item:semistable restrict}
        There exists an $\Sl$-semistable tensor $S \in \C^a \ot \C^b \ot \C^c$ such that $\tensor \geq S$.
        \item \label{item:semistable degeneration}
        There exists an $\Sl$-semistable tensor $S \in \C^a \ot \C^b \ot \C^c$ such that $\tensor \degengeq S$.
    \end{enumerate}
\end{lemma}
\begin{proof}
We first prove \ref{item:semistable} $\Rightarrow$ \ref{item:semistable restrict}.
There is a sequence $(A_i,B_i,C_i) \in \G$ where $i \in \N$, such that $(A_i,B_i,C_i) \cdot T$ converges to a tensor $S_0' \in \C^{a'} \ot \C^{b'} \ot \C^{c'}$ and
$\spec(\mu(S_0')) = ( u_a \sep u_b \sep u_c) \in \R^{a'} \times \R^{b'} \times \R^{c'}$.
Then by \cref{proposition:tensor moment polytope generic restriction}, there exists a restriction $S_0 \in \C^a \ot \C^b \ot \C^c$ of $S_0'$ which contains $( u_a \sep u_b \sep u_c) \in \R^a \times \R^b \times \R^c$ in its moment polytope.
By \cref{corollary:uniform marginals}, $S_0$ is $\SL$-semistable.
Let $(P_1,P_2,P_3)$ be the restriction map satisfying $(P_1,P_2,P_3) \cdot S_0' = S_0$.
Then $S_i = (P_1A_i,P_2B_i,P_3C_i) \cdot T \in \C^a \ot \C^b \ot \C^b$  is a sequence of tensors converging to $S_0$.
The set of $\Sl$-semistable tensors in $\C^a \ot \C^b \ot \C^c$ is Euclidean-open.
To see this, note that the tensors that are not $\Sl$-semistable are precisely the simultaneous zero-set of a finite set of polynomials (\cref{proposition:tensor null-cone}), hence Euclidean-closed, and its complement is Euclidean-open.
We conclude there must exists an $i$ such that $S = (P_1A_i,P_2B_i,P_3C_i) \cdot T$ is $\Sl$-semistable.

The implication \ref{item:semistable restrict} $\Rightarrow$ \ref{item:semistable degeneration} follows from the general fact that any restriction is a degeneration.

The implication \ref{item:semistable degeneration} $\Rightarrow$ \ref{item:semistable} follows from \cref{proposition:degeneration monotone}, which gives that $\Delta(S) \subseteq \Delta(T)$ because $S \degenleq T$, and \cref{corollary:uniform marginals}, which gives that $(u_a\sep u_b \sep u_c) \in \Delta(S)$.
\end{proof}

\subsection{Matrix pencils and polynomial multiplication tensors}
A special property of the space $\C^2 \ot \C^{c-1} \ot \C^c$ is that in here the $\Sl$-semistable tensors form a single $\G$-orbit $\G \cdot \polmul_c$, where $\polmul_c$ is given by %
\begin{align}
    \label{definition:matrix pencils tensor}
    \polmul_{c}
        &\coloneqq \sum_{i \in [2]} \sum_{j \in [c-1]} e_i \ot e_j \ot e_{i+j-1} \\
        &= e_1 \ot
        \left[\begin{array}{@{}c|c@{}}
            I_{c-1} & 0
        \end{array}\right]
        +
        e_2 \ot
        \left[\begin{array}{@{}c|c@{}}
            0 & I_{c-1}
        \end{array}\right],
\end{align}
where $\left[\begin{array}{@{}c|c@{}}
            I_{c-1} & 0
        \end{array}\right]$ denotes the concatenation of the $(c-1)\times (c-1)$ identity matrix with a zero column, and similarly for $\left[\begin{array}{@{}c|c@{}}
            0 & I_{c-1}
        \end{array}\right]$.
This then allows us to take $S = \polmul_c$ in \Cref{lemma:tensor semistability}.

We establish this property in three steps.
First, we remark that tensors in $\C^2 \ot \C^b \ot \C^c$ are known as matrix pencils \cite[Section~19]{burgisser1996algebraic}.
We use a result from the theory of matrix pencils \cite{pokrzywaMatrixPencils1986}, which states that the $\G$-orbit of $\polmul_c$ is dense. That is, $\overline{\G \cdot \polmul_c} = \C^2 \ot \C^{c-1} \ot \C^c$ (\cref{lemma:matrix pencils}).
We then prove that the $\Sl$-orbit of $\polmul_c$ is closed (\Cref{lemma:S_c closed SL orbit}).
This allows us to show that all tensors in the boundary $\overline{\G \cdot \polmul_c} \setminus (\G \cdot \polmul_c)$ are $\Sl$-unstable (\Cref{lemma:orbit border}).
Combining these results, all $\Sl$-semistable tensors in $\C^2 \ot \C^{c-1} \ot \C^c$ must indeed lie in $\G \cdot \polmul_c$.

\begin{lemma}[{\cite[Section ``Minimal Pencils'', p.\ 119]{pokrzywaMatrixPencils1986}}]
    \label{lemma:matrix pencils}
    Let $b \neq c$.
    Then there exists a $\G$-orbit in $\C^2 \ot \C^b \ot \C^c$ that is dense. %
    When $b = c - 1$, this is the~$\G$-orbit of~$\polmul_c$.
\end{lemma}

\begin{lemma}%
    \label{lemma:S_c closed SL orbit}
    The $\Sl$-orbit of $\polmul_c$ is closed.
\end{lemma}
\begin{proof}
    Acting with suitable diagonal matrices on the second and third components, one can show that there exists $C \in \R$ such that $C \polmul_{c}$ is~$\Sl$-equivalent to
    \begin{align}
        \polmul_{c}' \coloneqq
        e_1 \ot
        \scalebox{0.8}{$
        \left[\begin{array}{cccc|c}
                 \sqrt{c-1} &        &          &          & 0      \\
                          & \ddots &          &          & 0      \\
                          &        & \sqrt{2} &          & \vdots \\
                          &        &          & \sqrt{1} & 0
            \end{array}\right]$}
        +
        e_2 \ot
        \scalebox{0.8}{$
        \left[\begin{array}{c|cccc}
             0      & \sqrt{1}   &          &        &          \\
             0      &            & \sqrt{2} &        &          \\
             \vdots &            &          & \ddots &          \\
             0      &            &          &        & \sqrt{c-1}
        \end{array}\right]$}.
    \end{align}
    Another straightforward computation shows that $\mu(\polmul_{c}') = (I_2,I_{c-1}/(c-1),I_c/c)$.
    Therefore, because $\Sl \cdot C\polmul_c$ contains this tensor, the Kempf--Ness theorem (\cref{theorem:tensor kempf-ness}) then implies that $\Sl \cdot C\polmul_{c}$ is closed, and hence also $\Sl \cdot \polmul_c$ is closed.
\end{proof}

The proof of the following lemma is essentially the same argument as in \cite[Proposition~3.10]{burgisserFundamentalInvariants2017},
combined with \cref{proposition:tensor null-cone}.
We provide it for convenience of the reader.

\begin{lemma}[{\cite{burgisserFundamentalInvariants2017}}]
    \label{lemma:orbit border}
    Let $S \in \C^a \ot \C^b \ot \C^c$ have a closed $\Sl$-orbit.
    Then every element in the boundary $\overline{\G \cdot S} \setminus (\G \cdot S)$ is $\Sl$-unstable.
\end{lemma}
\begin{proof}
Let $f \in \C[V]$ be a homogeneous $\Sl$-invariant polynomial of degree $m > 0$.
We show $f$ evaluates to zero on the boundary, after which \cref{proposition:tensor null-cone} gives the desired result.
Let $R \in \overline{\G \cdot S} \setminus (\G \cdot S)$.
Take $g_i \in \G$ such that $\lim_{i \to \infty} g_i \cdot S = R$.
There are $h_i \in \Sl$ and $(t_{i,1},t_{i,2},t_{i,3}) \in \C^3$ such that $g_i = (t_{i,1}I,t_{i,2}I,t_{i,3}I) h_i$, where we multiply the 3-tuples entry-wise. Let $t_i = t_{i,1}t_{i,2}t_{i,3}$.
Then $f(g_i \cdot S) = f\big(t_i(h_i \cdot S)\big) = t_i^{m} f(h_i \cdot S) = t_i^m f(S)$.
Because $g_i \cdot S$ converges to $R$ and by continuity of $f$, it follows that $|t_i|$ must converge as well.
Hence $t_i$ is a bounded sequence, and we may pass to a subsequence such that the limit $t = \lim_{i \to \infty} t_i$ exists.
Then~$t$ equals 0, as otherwise $R / t = \lim_{i \to \infty} g_i \cdot S/t_i
= \lim_{i\to\infty} h_i \cdot S \in \overline{\Sl \cdot S} = \Sl \cdot S$, and as a result $R \in \G \cdot S$, which is a contradiction.
It follows that $t = 0$ and hence $f(R) = \lim_{i\to\infty} f(g_i \cdot S) = \lim_{i\to\infty} t_i^m f(S) = 0$.
\end{proof}

We combine the above lemmas.
\begin{corollary}
    \label{corollary:S_c all semistable}
    The $\G$-orbit of $\polmul_c$ consists of all the $\Sl$-semistable tensors in $\C^2 \ot \C^{c-1}\ot \C^c$.
\end{corollary}
\begin{proof}
    By \cref{lemma:S_c closed SL orbit}, the $\Sl$-orbit of $\polmul_c$ is closed, and hence~$\polmul_c$ is also~$\Sl$-semistable.
    Then \cref{lemma:orbit border} tells us that the boundary of~$\G \cdot \polmul_c$ contains only $\Sl$-unstable tensors.
    By \cref{lemma:matrix pencils}, the $\G$-orbit of $\polmul_c$ is dense.
    Hence all tensors outside of the $\G$-orbit lie on its boundary, and are therefore $\Sl$-unstable.
    We conclude $\G \cdot \polmul_c$ contains all $\Sl$-semistable tensors.
    Moreover, as~$\polmul_c$ itself is~$\Sl$-semistable, every tensor in~$\G \cdot \polmul_c$ is semistable (as they are in the~$\Sl$-orbits of some scalar multiples of~$\polmul_c$).
\end{proof}

The tensor $\polmul_c$ has a computational interpretation, namely as the structure tensor of multiplication of univariate polynomials with degrees~$1$ and~$c-1$, respectively.
Hence it is a special case of the structure tensors describing multiplications of two polynomials of given degrees.

\begin{definition}[Polynomial multiplication tensor]
\label{definition:polynomial multiplication tensors}
The structure tensor $\polmul_{a,b} \in \C^a\ot\C^b\ot\C^{a+b-1}$ describing the multiplication of two univariate polynomials, the first of degree $a-1$ and the second of degree $b-1$, is given by
\begin{align}
    \polmul_{a,b} &\coloneqq \sum_{i \in [a]} \sum_{j \in [b]} e_i \ot e_j \ot e_{i+j-1} \\
        &=
        e_1 \ot
        \left[\begin{array}{@{}c|c@{}}
            I_{b} & 0_{a-1}
        \end{array}\right]
        +
        e_2 \ot
        \left[\begin{array}{@{}c|c|c@{}}
            0_1 & I_{b} & 0_{a-2}
        \end{array}\right]
        +
        \cdots
        +
        e_a \ot
        \left[\begin{array}{@{}c|c@{}}
            0_{a-1} & I_{b}
        \end{array}\right],
\end{align}
where $0_{t} \in \C^{b\times t}$ denotes the $b\times t$ zero matrix.
\end{definition}

Clearly, we have $\polmul_{c} = \polmul_{2,c-1}$.
We shall use the tensors~$\polmul_{a,b}$ in~\cref{subsec:border-subrank}.
In particular, we will use that they have low tensor rank:\footnote{In fact, the tensor rank of $\polmul_{a,b}$ is minimal for concise tensors in $\C^a \ot \C^b \ot \C^{a+b-1}$.}
\begin{lemma}[{\cite[Proposition~14.47]{burgisser1996algebraic}}]\label{lem:ranksab}
$\unit{a + b - 1} \geq \polmul_{a,b}$.
\end{lemma}

\subsection{The minrank of a tensor}

For the proof of \cref{theorem:matrix multiplication} we will use the notion of \emph{minrank} of a tensor (see also \cite{briet_et_al:LIPIcs.ITCS.2024.20, blaser2019varietymembershiptestingalgebraic}).
To define the minrank for a tensor $T \in \C^a\ot\C^b\ot\C^c$, we ``slice'' $T$ into an $a$-tuple of $b \times c$ matrices $\big( \big[T_{1,j,k}\big]_{j,k}$, $\big[T_{2,j,k}\big]_{j,k}$, $\ldots$ $\big[T_{a,j,k}\big]_{j,k} \big)$.
Then the minrank is the smallest nonzero matrix rank of any linear combination of these slices.
For 3-tensors there are three ways to do this, one for each index, and hence $T$ has three different ``minranks''. We will only need the one with above slicing, which we simply denote by $\minrank(T)$.
In other words:
\begin{definition}%
  \label{def:minrank}
  Let~$T \in \C^a \ot \C^b \ot \C^c$ be a non-zero tensor.
  Then we define its \emph{minrank} by
  \[
  \minrank(T) \coloneqq
      \min \bigl( \bigl\{\rank\! \big( (\beta \ot I_b \ot I_c) T \big) \ \big|\
      \beta \in \C^{1 \times a}\bigr\} \setminus \{0\} \bigr).
  \]
\end{definition}

\begin{lemma}
  \label{prop:minrank monotonicity}
  Suppose $T_1, T_2, T_3, \ldots \in \C^{a} \ot \C^{b} \ot \C^{c}$ converge to a concise $T \in \C^{a} \ot \C^{b} \ot \C^{c}$.
  Then~$\minrank(T) \leq \liminf_{i\to\infty} \minrank(T_i)$.
\end{lemma}
\begin{proof}
  Let $r = \liminf_{i\to\infty} \minrank(T_i)$.
  Since minrank is integer valued, we may pass to a subsequence such that~$\minrank(T_i) = r$ for every~$i$.
  Let~$\beta_i \in \C^{1 \times a}$ be such that~$\minrank(T_i) = \rank((\beta_i \ot I_{b} \ot I_c) T_i)$ for every~$i$.
  Replace the~$\beta_i$'s by~$\beta_i / \norm{\beta_i}$, so that their norm is~$1$.
  By compactness of the unit sphere, we may pass to a subsequence of $(T_i,\beta_i)_i$ such that~$\beta_i \to \beta$ for some $\beta \in \C^{1 \times a}$ with~$\norm{\beta}=1$.
  Define $f \colon \C^{1\times a} \times \C^{a\times b\times c} \to \C^{b \times c} \colon (\gamma, S) \mapsto (\gamma \ot I_b \ot I_c) S$.
  Then we find by continuity of $f$ that
  \begin{align*}
    (\beta \ot I_a \ot I_c)T 
    = f(\beta,T) 
    = f\big(\lim_{i\to\infty} (\beta_i,T_i)\big)
    = \lim_{i\to\infty} f\big( (\beta_i,T_i)\big)
    = \lim_{i\to\infty} 
    (\beta_i \ot I_a \ot I_c)T_i. 
  \end{align*}
  Hence, $\rank((\beta \ot I_b \ot I_c) T) \leq \lim_{i\to\infty} \rank((\beta_i \ot I_b \ot I_c) T_i)$, since matrix rank cannot go up in the limit. 
  Since  $\rank((\beta_i \ot I_b \ot I_c) T_i) = r$ for all $i$, we find that $\rank((\beta \ot I_b \ot I_c) T) \leq r$.
  Moreover~$(\beta \ot I_b \ot I_c) T \neq 0$ by conciseness of~$T$ and since $\beta$ is nonzero, so that its rank is not~$0$, and hence~$\minrank(T) \leq \rank((\beta \ot I_b \ot I_c) T)$.
  Therefore~$\minrank(T) \leq r$.
\end{proof}

For the proof of~\cref{theorem:matrix multiplication} we will also use the following:
\begin{lemma}\label{lem:ingr-2}
$\minrank(\polmul_{a,b}) = b$.
\end{lemma}
\begin{proof}
  Every nonzero matrix $(\beta \ot I_b \ot I_{a+b}) \polmul_{a,b}$, where $\beta \in \C^a$, has a $b \times b$ submatrix that is upper triangular with non-zero diagonal entries.
  These matrices have rank~$b$ and hence $\minrank(\polmul_{a,b}) = b$.
\end{proof}

\subsection{Minrank of degenerations of matrix multiplication}
The next important ingredient for~\cref{theorem:matrix multiplication} is the following lemma:
\begin{lemma}\label{lem:ingr-1}
Let $T \in \C^a \ot \C^b \ot \C^c$ be a non-zero and concise tensor. If $\MM_{n} \degengeq T$, then we have $\minrank(T) \leq n (n- \floor{\sqrt{a-1}})$.
\end{lemma}
We will use the following lemma, which follows from the fact that projective varieties of complementary dimension must intersect and a computation of the dimension of the variety of $n \times n$ matrices with rank at most $r$, see, for instance,~{\cite[Proposition~6]{cubitt2008dimensionOfSubspacesBoundedSchmidtRank}} or \cite{MR2645077}.

\begin{lemma}%
\label{corollary:smallest rank in matrix subspace}
    Every $d$-dimensional subspace of $n \times n$ matrices contains a nonzero matrix of rank at most $n -\lfloor\sqrt{d-1}\rfloor$.
\end{lemma}
\begin{proof}[Proof of~\cref{lem:ingr-1}]
Suppose $\MM_{n} \geq T$ with $T$ concise. Then there exist linear maps~$A,B,C$ such that~$(A \ot B \ot C) \MM_{n} = T$.
Recall that
$
    \MM_{n}
    = \sum_{i,j} e_{i,j} \ot \sum_k e_{j,k} \ot e_{k,i}.
$
Let $E_{j,i} = e_j \otimes e_i$, which we may think of as an $n \times n$ matrix. Let $E_{j,i} \kron I_n$ denote the matrix Kronecker product with the $n \times n$ identity matrix. Then after a permutation of basis elements we may write
\begin{align}
    \MM_{n}
    = \sum_{i,j} e_{i,j} \ot (E_{j,i} \kron I_n).
\end{align}
By applying the linear map $A\colon \C^{n^2} \to \C^a$ we take linear combinations of these slices:
\begin{align}
    (A \ot I_{n^2} \ot I_{n^2}) \MM_{n}
    = \sum_{\ell=1}^a e_\ell \ot (M_{\ell} \kron I_n),
\end{align}
for some matrices $M_\ell \in \C^{n \times n}$ (namely, $M_\ell = \sum_{i,j} A_{\ell,(i,j)} E_{j,i}$).
Then $M_1, \ldots, M_a$ are linearly independent, as otherwise $(A \ot I_{n^2} \ot I_{n^2}) \MM_{n}$ and hence $T$ would not be concise.
By \cref{corollary:smallest rank in matrix subspace}, there exists $\beta \in \C^a \setminus \{0\}$ such that the rank of
$\sum_\ell \beta_\ell M_\ell \neq 0$ is at most $n - \lfloor\sqrt{a-1}\rfloor$.
It follows that the rank of $\sum_\ell \beta_\ell (M_{\ell} \kron I_n) = \big(\sum_\ell \beta_\ell M_\ell\big) \kron I_n$ is at most $n(n-\lfloor\sqrt{a-1}\rfloor)$.

Next, we observe that applying~$I \ot B \ot C$ corresponds to left- and right-multiplication of $\sum_\ell \beta_\ell (M_{\ell} \kron I_n)$ with~$B$ and~$C$, which can only make rank go down.
By conciseness of $T$, $(\beta A \ot B \ot C) \MM_{n} = (\beta \ot I \ot I)T$ is non-zero.
It follows that
\begin{align*}
    \minrank \big((A \ot B \ot C) \MM_{n}\big)
    & \leq
    \rank\big((\beta A \ot B \ot C) \MM_{n}\big)
    \\
    & \leq
    \rank\big((\beta A \ot I_{n^2} \ot I_{n^2}) \MM_{n}\big)
    \\
    & \leq
    n\big(n-\lfloor\sqrt{a-1}\rfloor\big).
\end{align*}
This yields~$\minrank(T) \leq n\big(n-\lfloor\sqrt{a-1}\rfloor\big)$.

We now extend the proof to concise~$T$ such that~$\MM_{n} \degengeq T$.
This implies that there exists a sequence of~$T_i$'s in the same format as~$T$ such that~$\MM_{n} \geq T_i$ and~$\lim_{i\to\infty} T_i = T$.
As conciseness is an open condition, we may pass to a subsequence such that the~$T_i$ are concise
The above proof shows that~$\minrank(T_i) \leq n (n - \lfloor\sqrt{a-1}\rfloor)$ for every~$i \geq 1$.
\Cref{prop:minrank monotonicity} then shows that~$\minrank(T) \leq n (n - \lfloor\sqrt{a-1}\rfloor)$.
\end{proof}

\subsection{Proof of \texorpdfstring{\cref{theorem:matrix multiplication}}{Theorem 3.1}}
Now that we have all the ingredients in place, we prove the moment polytope separation:
\begin{proof}[Proof of \cref{theorem:matrix multiplication}]
Let $T$ be any tensor.
By \Cref{lemma:tensor semistability}, we have $p_c = (u_2 \sep u_{c-1} \sep u_c) \in \Delta(T)$ if and only if there exists an $\SL$-semistable tensor $S \in \C^2 \ot \C^{c-1} \ot \C^c$ such that $T \geq S$.
By \Cref{corollary:S_c all semistable}, $S$ is in the $\GL$-orbit of $\polmul_c = \polmul_{2,c-1}$.
So $p_c \in \Delta(T)$ if and only if $T \geq \polmul_c$.

We have $p_c \in \Delta(\unit{c})$, because $\unit{c} \geq \polmul_c$ (\cref{lem:ranksab}).

To prove $p_c \notin \Delta(\MM_{n})$, we will show that~$\MM_{n} \not\degengeq \polmul_c$. 
Suppose $\MM_{n} \degengeq \polmul_c$. Then by \cref{lem:ingr-1}, $\minrank(\polmul_c) \leq n(n - 1)$. On the other hand, by \cref{lem:ingr-2}, $\minrank(\polmul_c) = c-1$. This contradicts the assumption that $n^2-n+1 < c$.
\end{proof}

\section{Extensions, applications and further results}
\subsection{Separation for iterated matrix multiplication and unit tensors}

We will prove \cref{cor:it-mat} using the projection relation between $\Delta(\IMM{n}{k})$ and $\Delta(\IMM{n}{k-1})$ given in \cref{lem:IMM polytope reduction}, which will allow us to reduce to the case $k=3$ (\cref{theorem:matrix multiplication}). %
First we prove a lemma that will help us translate specific restrictions from $\MM_{n}^k$ to $\MM_{n}^{k-1}$.

\begin{lemma}
\label{lemma:rank 1 restriction of MM restriction}
    Suppose $T = (A_1 \ot \cdots \ot A_k) \MM_{n}^k$ is a restriction such that  $\rank(A_k) = 1$.
    Then $T = S \ot w$ with $\MM_{n}^{k-1} \geq S$, where $S$ is a tensor of order $k-1$ and $w$ is a vector. %
\end{lemma}
\begin{proof}
    Because $A_k$ has rank 1, we may write it as $A_k = w u^*$ for some $u \in \C^{n^2}$ and vector $w$.
    Then we find that
    \begin{align*}
        T &= (A_1 \ot \dotsb \ot A_{k-1} \ot w u^*) \IMM{n}{k}
        \\&= \sum_{i_1,\ldots,i_k} A_1 e_{i_1,i_2} \ot \cdots \ot A_{k-1} e_{i_{k-1},i_{k}} \ot wu^* e_{i_{k},i_{1}}
        \\&= \sum_{i_1,\ldots,i_k} A_1 e_{i_1,i_2} \ot \cdots \ot A_{k-1} e_{i_{k-1},i_{k}} \ot w u_{i_{k},i_{1}}
        \\&= \sum_{i_1,\ldots,i_{k-1}} A_1 e_{i_1,i_2} \ot \cdots \ot A_{k-1} \sum_{i_{k}} u_{i_{k},i_{1}} e_{i_{k-1},i_{k}}  \ot w
        \\&= \sum_{i_1,\ldots,i_{k-1}} A_1 e_{i_1,i_2} \ot \cdots \ot A_{k-1} U e_{i_{k-1},i_{1}} \ot w,
    \end{align*}
    where $U \colon \C^{n^2} \to \C^{n^2}$ is the linear map that maps
    $e_{i_{k-1},i_{1}}$ to $\sum_{i_{k}} u_{i_k,1} e_{i_{k-1},i_k}$.
    We find that $T = S \ot w$ with $S = (A_1 \ot \cdots \ot A_{k-1}U) \MM_{n}^{k-1}$, and we are done.
\end{proof}
We use \cref{lemma:rank 1 restriction of MM restriction} to relate the moment polytopes of $\MM_{n}^{k-1}$ and $\MM_{n}^k$.
\begin{lemma}
    \label{lem:IMM polytope reduction}
    Let $k \geq 3$.
    Then
    \[
    \{q \in \Delta(\IMM{n}{k}) \mid q_k = (1, 0, \dotsc, 0) \} = \Delta(\IMM{n}{k-1}) \times \{(1, 0, \dotsc, 0)\}.
    \]
\end{lemma}
\begin{proof}
    $\subseteq$. Let $q \in \Delta(\MM_{n}^{\smash{k}})$ be such that $q_k = (1,0,\ldots,0)$. Let $T \in \overline{\G \cdot \IMM{n}{k}}$ be such that $\spec \mu(T) = q$.
    Then the $k$-th flattening of $T$ has rank 1, so $T = S \ot w$ for some $(k-1)$-tensor $S$ and $w \in \C^{\smash{n^2}}$.
    Let $T_i = (A_1^{\smash{(i)}} \ot \ldots \ot A_{k}^{\smash{(i)}}) \IMM{n}{k}$ be restrictions such that $\lim_{i\to\infty} T_i = T$. 
    Then also $(I \ot \cdots \ot I \ot ww^*) T_i$ converges to $T$.
    Because $ww^* A_{k}^{\smash{(i)}}$ has rank 1, we may apply \cref{lemma:rank 1 restriction of MM restriction} for each $i$.
    We find that $T_i = S_i \otimes w_i$ with $S_i \leq \IMM{n}{k-1}$ and $w_i \in \C^{\smash{n^2}}$. 
    Therefore, 
    \begin{align*}
        \spec \mu(T_i) 
        = \bigl(\spec \mu(S_i)
        ,  (1,0,\ldots,0)\bigr) 
        \ \in\ \Delta(\IMM{n}{k-1}) \times \{(1,0,\ldots,0)\}
    \end{align*}
    for all $i$.
    Then $q = \lim_{i \to \infty} \spec \mu(T_i) \in \Delta(\IMM{n}{k-1}) \times \{(1,0,\ldots,0)\}$ because of closedness.

    $\supseteq$. %
    The tensor $\IMM{n}{k}$ restricts to $\IMM{n}{k-1} \ot e_1$ by applying to the $k$th factor of $\IMM{n}{k}$ the linear map that maps $e_{i,i}$ to~$\frac1n e_1$ for all~$i \in [n]$, and~$e_{i,j}$ to~$0$ if~$i \neq j$.
    Thus $\Delta(\IMM{n}{k}) \supseteq \Delta(\IMM{n}{k-1} \ot e_1)$.
    Note that $\Delta(\IMM{n}{k-1} \ot e_1) = \Delta(\IMM{n}{k-1}) \times \{(1,0,\ldots,0)\}$.
\end{proof}

\begin{proof}[Proof of \cref{cor:it-mat}]
    By \cref{theorem:matrix multiplication}, $p_c = ( u_2 \sep u_{c-1} \sep u_c ) \in \Delta(\unitk{c}{3})$ and $p_c \not\in\Delta(\IMM{n}{3})$.
    The tensor $\unitk{c}{k}$ restricts to $\unitk{c}{3} \otimes e_1^{\otimes k-3}$.
    Therefore, $( u_2 \sep u_{c-1} \sep u_c \sep u_1 \sep \cdots \sep u_1) \in \Delta(\unitk{c}{k})$.
    Suppose that $( u_2 \sep u_{c-1} \sep u_c \sep u_1 \sep \cdots \sep u_1) \in \Delta(\IMM{n}{k})$. Then $( u_2 \sep u_{c-1} \sep u_c) \in \Delta(\IMM{n}{3})$ by~\cref{lem:IMM polytope reduction}, which is a contradiction.
\end{proof}

\subsection{Border subrank of matrix multiplication}\label{subsec:border-subrank}

We now revisit the proof techniques of \cref{theorem:matrix multiplication} and push them to prove the border subrank upper bound $\bordersubrank\big(\MM_{n}\big) \leq \lceil \tfrac34 n^2\rceil$, \cref{theorem:subrank bound}.
The core idea is to show that the matrix multiplication tensor cannot degenerate to polynomial multiplication tensors $\polmul_{a,b}$ (\Cref{definition:polynomial multiplication tensors}) for certain choices of $a$ and $b$.
Since $\unit{a+b-1} \degengeq \polmul_{a,b}$, this prevents $\MM_{n} \degengeq \unit{a+b-1}$ by transitivity of degeneration. Hence $\bordersubrank\big(\MM_{n}\big) < a + b-1$.
For the optimal choices of $a$ and $b$ we obtain $a+b-1 = \lceil \tfrac34 n^2\rceil + 1$. %

\begin{proof}[Proof of \cref{theorem:subrank bound}
(using polynomial multiplication tensors)]

Let $a,b \in \N$.
Suppose that we have $b > n(n-\lfloor\sqrt{a-1}\rfloor)$.
Then we claim that $\MM_{n} \not\degengeq \unit{a+b-1}$.
Since $\unit{a+b-1} \geq \polmul_{a,b}$ (\cref{lem:ranksab}), it suffices to prove that $\MM_{n} \not\degengeq \polmul_{a,b}$.
Suppose~that $\MM_{n} \degengeq \polmul_{a,b}$.
Then by \cref{lem:ingr-1} we have  %
$\minrank(\polmul_{a,b}) \leq n (n- \floor{\sqrt{a-1}})$.
By \cref{lem:ingr-2},
$\minrank(\polmul_{a,b}) = b$,
so we find that $b \leq n(n-\lfloor\sqrt{a-1}\rfloor)$, which is a contradiction.

It remains to find $a,b\in \N$ such that we have $b > n(n-\lfloor\sqrt{a-1}\rfloor)$ and $a + b - 1 = \lceil \tfrac34n^2 \rceil+1$. Indeed, then $\MM_{n} \not\degengeq \unit{\lceil (3/4) n^2 \rceil+1}$, so $\bordersubrank(\MM_{n}) \leq \lceil \tfrac34 n^2 \rceil$.

Suppose that $n$ is even. Let $a = \tfrac14 n^2 + 1$ and $b = \tfrac12 n^2+1$.
Then $\sqrt{a-1} = \tfrac12 n = \lfloor \tfrac12 n \rfloor$, and we indeed have $b > \tfrac12 n^2 = n(n-\lfloor\sqrt{a-1}\rfloor)$
and $a+b-1-1 = \tfrac34n^2 = \lceil \tfrac34n^2 \rceil$.

Suppose that $n$ is odd. Let $a = \tfrac14 (n-1)^2 + 1$ and $b = \tfrac12 n(n+1)+1$.
Then $\sqrt{a-1} = \tfrac12 (n-1) = \lfloor \tfrac12 (n-1)\rfloor$, and we indeed have $b > \tfrac12 n^2 + \frac12 n = n(n-\lfloor\sqrt{a-1}\rfloor)$
and
$a+b-1-1 = \tfrac34n^2 + \tfrac14 = \lceil \tfrac34n^2 \rceil$. %
\end{proof}

\begin{remark}
We note that the tensor $\polmul_{a,b}$ is $\Sl$-semistable (extending \cref{corollary:S_c all semistable}).
This can be deduced using the \emph{semistability test} as given in \cite[Lemma 23]{christandl2023weightedSliceRank},
as adapted from \cite[Corollary~5.1]{kempf1978instability}.
It states that for any group $\Gamma$ and irreducible $\Gamma$-representations $V_1, V_2$ and $V_3$, whenever
 $S \in V_1 \ot V_2 \ot V_3$ is non-zero and invariant under the diagonal action of $\Gamma$, then $S$ is $\Sl$-semistable.
We consider the group $\Gamma = \SL_2$.
Then $\Gamma$ acts by basis transformations on the space $\C[x,y]_d$ of homogeneous degree $d$ polynomials in variables $x$ and $y$. This is an irreducible representation.
We map from $\C[x,y]_d$ to the univariate polymomials of degree at most $d$ by setting $y = 1$. In the other direction we homogenize the univariate polynomials in $x$ using $y$.
This identifies $\C[x,y]_d \cong \C^{d+1}$. Using this identification, $\polmul_{a,b}$ describes the bilinear multiplication map $\C[x,y]_{a-1} \times \C[x,y]_{b-1} \to \C[x,y]_{a+b-2}$.
This map is $\Gamma$-equivariant.
Hence $\polmul_{a,b} \in \C[x,y]_{a-1} \ot \C[x,y]_{b-1} \ot \C[x,y]_{a+b-2}^*$ is $\Gamma$-invariant. The semistability test then implies~$\polmul_{a,b}$ is $\Sl$-semistable.

However, unlike for $\polmul_c$ (\cref{corollary:S_c all semistable}), the tensor $\polmul_{a,b}$ for $a \geq 3$ is not necessarily the only $\Sl$-semistable tensor of its shape.
Indeed, using tensor scaling \cite{burgisser2018tensorScaling}, we observe numerically for~$n \in \{4, \ldots, 13\}$ and the optimal choices of~$a,b$ such that~$\MM_{n} \not\degengeq\polmul_{a,b}$ from the proof of \cref{theorem:subrank bound}, that $(u_a,u_b,u_{a+b-1})$ is in fact an element of $\Delta(\MM_{n})$.
By \cref{lemma:tensor semistability}, this means~$\MM_{n}$ must restrict to some semistable tensor of shape $a \times b \times (a+b-1)$ that is not equivalent to $\polmul_{a,b}$.
Thus this analysis does not directly lead to an extension of \cref{theorem:matrix multiplication}.
\end{remark}

\subsection{Asymptotic restriction does not imply moment polytope inclusion}\label{subsec:asymp-restr}
\label{subsection:asymptotic degeneration monotone}
We prove \cref{thm:moment-polytope-not-asymptotic-restriction-monotone} by giving two examples, \cref{example:asymp-restrict-but-no-polytope-inclusion-mm2-vs-4} and \cref{example:asymp-restrict-but-no-polytope-inclusion-Tdet-vs-3}.
The first example uses matrix multiplication tensors.
\begin{example}
    \label{example:asymp-restrict-but-no-polytope-inclusion-mm2-vs-4}
    We take $S = \unit{m^2}$ and $T = \MM_m$.
    Strassen~\cite[Equation 4.6]{strassen1988asymptotic} proved that $\unit{m^2} \asympleq \MM_m$.
    On the other hand, our \cref{theorem:matrix multiplication} gives
    $\Delta(\unit{m^2}) \not\subseteq \Delta(\MM_m)$.
\end{example}

The second example will use the unique non-zero skew-symmetric tensor in $\Lambda^3(\C^3) \subseteq \C^3 \ot \C^3 \ot \C^3$. We will need some ingredients to discuss it.
\begin{proposition}
    \label{prop:mm-ab1-res-deg}
    Let~$a,b \geq 1$.
    The orbit $\GL\cdot\MM_{a,1,b}$ is dense in~$\C^a \ot \C^b \ot \C^{ab}$ and equals the subset of $\Sl$-semistable tensors in $\C^a \ot \C^b \ot \C^{ab}$.
\end{proposition}
\begin{proof}
    Let $S \in \C^a \ot \C^b \ot \C^{ab}$.
    We may write $S = \sum_{i,k=1}^a \sum_{j,\ell=1}^b S_{i,j,(k,\ell)} e_i \ot e_j \ot e_{k,\ell}$ for some coefficients $S_{i,j,(k, \ell)}$.
    Let $M$ be the $ab \times ab$ matrix with coefficients $M_{(k,\ell), (i,j)} = S_{i,j,(k,\ell)}$. %
    Then
    $S = \sum_{i,j} e_i \ot e_j \ot M e_{i,j} = (I_a \ot I_b \ot M)\MM_{a,1,b}$.
    Therefore, $S$ is a restriction of~$\MM_{a,1,b}$. Thus~$S$ is also a degeneration of~$\MM_{a,1,b}$. We conclude that~$\overline{\G \cdot \MM_{a,1,b}} = \C^a \ot \C^b \ot \C^{ab}$.

    A straightforward computation shows that $\mu(\MM_{a,1,b}) = (I_a/a,I_b/b,I_{ab}/ab)$.
    By the Kempf--Ness theorem (\cref{theorem:tensor kempf-ness}), this means $\MM_{a,1,b}$ is semistable and has closed $\Sl$ orbit.
    In particular, all tensors in $\G \cdot \MM_{a,1,b}$ are semistable.
    Moreover, we may apply \cref{lemma:orbit border} to obtain that all tensors in the border of $\G \cdot \MM_{a,1,b}$ are unstable. Because this orbit is dense, these are all tensors outside the orbit, and we are done.
\end{proof}
\begin{corollary}
    $\Delta(\MM_{a,1,b}) = \kronpol{a}{b}{ab}$.
\end{corollary}
\begin{proof}
    \Cref{prop:mm-ab1-res-deg} implies that $\overline{\G \cdot \MM_{a,1,b}} = \C^a\ot\C^b\ot\C^{ab}$, after which the definition of $\Delta$ (\cref{definition:moment polytope}) directly implies the result.
\end{proof}

\begin{corollary}
    \label{cor:mm-ab1-point-inclusion}
    Let $a,b \geq 1$.
    Then for any tensor $T$
    the following are equivalent:
    \begin{enumerate}[label=\upshape(\arabic*)]
    \item $\MM_{a,1,b} \leq T$,
    \item $\MM_{a,1,b} \degenleq T$,
    \item $\Delta(\MM_{a,1,b}) \subseteq \Delta(T)$,
    \item $(\,u_a \sep u_b \sep u_{ab}\,) \in \Delta(T)$.
    \end{enumerate}
\end{corollary}
\begin{proof}
  This follows from combining~\cref{prop:mm-ab1-res-deg} with \cref{lemma:tensor semistability}.
\end{proof}

\begin{example}
    \label{example:asymp-restrict-but-no-polytope-inclusion-Tdet-vs-3}
    Let $S = \MM_{1,1,3}$ and $T = e_1 \wedge e_2 \wedge e_3 = e_{1,2,3} - e_{1,3,2} + e_{2,3,1} - e_{2,1,3} + e_{3,1,2} - e_{3,2,1}$, where we write $e_{i,j,k} = e_i \ot e_j \ot e_k$.

    We claim $\Delta(\MM_{1,1,3}) \not\subseteq \Delta(T)$.
    Let $\maxrank(T)$ be the largest rank of any element in the matrix subspace $\{(\beta \otimes I \otimes I) \cdot T \mid \beta \in \C^{1\times n}\}$. Then $\maxrank(T)$ equals the largest $r$ such that $T \geq \MM_{1,1,r}$.
    We prove $\maxrank(T) = 2$.
    Indeed, %
    viewed as a matrix,
    \begin{equation*}
        (\beta \ot I \ot I) \cdot T
        = \begin{bmatrix}
            0        & \beta_3   & -\beta_2 \\
            -\beta_3 & 0         & \beta_1 \\
            \beta_2 & - \beta_1 & 0
        \end{bmatrix}.
    \end{equation*}
    The determinant of this matrix equals 0 for every $\beta$.
    Therefore, $\maxrank(T) \leq 2$.
    It follows that $T \not\geq \MM_{1,1,3}$ and thus $\Delta(\MM_{1,1,3}) \not\subseteq \Delta(T)$ by \cref{cor:mm-ab1-point-inclusion}.

    We claim $T \gtrsim \unit{3}$. Then $T \gtrsim \MM_{1,1,3}$ since $\unit{3} \geq \MM_{1,1,3}$.
    The inequality $T \gtrsim \unit{3}$ follows from the characterization of asymptotic subrank for tight tensors of Strassen \cite[Lemma~5.1, Proposition~5.4]{strassen1991} (see also
    \cite[Theorem~4.4, Corollary~4.5]{christandlUniversalPointsAsymptotic2021}).
    For any tensor $S \in \C^{n_1} \ot \C^{n_2} \ot \C^{n_3}$ the support $\supp(S)$
    is called tight if there exist injective maps $f_\ell \colon [n_\ell] \to \Z$ such that $f_1(i) + f_2(j) + f_3(k) = 0$ for every $(i,j,k) \in \supp(S)$.
    The characterization states: Suppose $\supp(S)$ is tight and let $r = \max_p \min\{ 2^{H(p_1)}, 2^{H(p_2)}, 2^{H(p_3)}\}$, where $p$ goes over all probability distributions on $\supp(T)$, and $p_1, p_2, p_3$ denote its marginal distributions. Then $S^{\boxtimes n} \geq \unit{r^{n - o(n)}}$.
    Note that $\supp(T)$ is tight by taking $f_\ell(1) = f_\ell(2) = 1$ and $f_\ell(3) = -2$.
    Choosing $p(i,j,k) = 1/6$ for all $(i,j,k) \in \supp(T)$ gives $r \geq 3$.
    We conclude $T \gtrsim \<3>$.
\end{example}

\paragraph*{Acknowledgements}

MvdB, VL, and MW acknowledge support by the European Research Council (ERC Grant Agreement No.~101040907).
MvdB also acknowledges financial support by the Dutch National Growth Fund (NGF), as part of the Quantum Delta NL visitor programme.
MC and HN acknowledge financial support from the European Research Council (ERC Grant Agreement No.~818761), VILLUM FONDEN via the QMATH Centre of Excellence (Grant No.~10059) and the Novo Nordisk Foundation (grant NNF20OC0059939 `Quantum for Life'). MC also thanks the National Center for Competence in Research SwissMAP of the Swiss National Science Foundation and the Section of Mathematics at the University of Geneva for their hospitality. Part of this work was completed while MC was Turing Chair for Quantum Software, associated to the QuSoft research center in Amsterdam, acknowledging financial support by the Dutch National Growth Fund (NGF), as part of the Quantum Delta NL visitor programme.
HN also acknowledges support by the European Union via a ERC grant (QInteract, Grant No.~101078107).
MW also acknowledges the Deutsche Forschungsgemeinschaft (DFG, German Research Foundation) under Germany's Excellence Strategy - EXC\ 2092\ CASA - 390781972, the BMBF (QuBRA, 13N16135; QuSol, 13N17173) and the Dutch Research Council (NWO grant OCENW.KLEIN.267).
JZ was supported by NWO Veni grant VI.Veni.212.284.
Views and opinions expressed are those of the author(s) only and do not necessarily reflect those of the European Union or the European Research Council Executive Agency. Neither the European Union nor the granting authority can be held responsible for them.

\bibliographystyle{alphaurl}
\bibliography{references}

\end{document}